\documentclass[12pt]{article}
\usepackage{amsmath, amsthm, amssymb, fullpage, graphicx, natbib}
\newtheorem{assumption}{Assumption}

\newtheorem{proposition}{Proposition}
\newtheorem{theorem}{Theorem}

\title{Does Low Spoilage Under Cold Conditions Foster Cultural Complexity During the Foraging Era? \\ A Theoretical and Computational Inquiry}
\author{Minhyeok Lee}
\date{}

\begin{document}
\maketitle

\begin{abstract}
Human cultural complexity did not arise in a vacuum. Scholars in the humanities and social sciences have long debated how ecological factors, such as climate and resource availability, enabled early hunter-gatherers to allocate time and energy beyond basic subsistence tasks. This paper presents a formal, interdisciplinary approach that integrates theoretical modeling with computational methods to examine whether conditions that allow lower spoilage of stored food, often associated with colder climates and abundant large fauna, could indirectly foster the emergence of cultural complexity. Our contribution is twofold. First, we propose a mathematical framework that relates spoilage rates, yield levels, resource management skills, and cultural activities. Under this framework, we prove that lower spoilage and adequate yields reduce the frequency of hunting, thus freeing substantial time for cultural pursuits. Second, we implement a reinforcement learning simulation, inspired by engineering optimization techniques, to validate the theoretical predictions. By training agents in different $(Y,p)$ environments, where $Y$ is yield and $p$ is the probability of daily spoilage, we observe patterns consistent with the theoretical model: stable conditions with lower spoilage strongly correlate with increased cultural complexity. While we do not claim to replicate prehistoric social realities directly, our results suggest that ecologically stable niches provided a milieu in which cultural forms could germinate and evolve. This study, therefore, offers an integrative perspective that unites humanistic inquiries into the origins of culture with the formal rigor and exploratory power of computational modeling.
\end{abstract}

\section{Introduction}
The origin and early expansion of human cultural complexity have occupied a central place in debates among archaeologists, anthropologists, and historians of prehistory. Since the early twentieth century, scholars have recognized that certain ecological conditions might have created more opportunities for artistic, symbolic, and social innovation. Yet, these arguments often relied on speculative, albeit insightful, narratives rather than robust theoretical or computational foundations. With the advent of more nuanced analytical tools, recent scholarship has begun to rigorously test the links between environmental stability, resource availability, and cultural development. By leveraging computational modeling alongside ethnographic and archaeological data, contemporary research aims to understand not just the presence of symbolic artifacts or social institutions, but also the underlying ecological and technological preconditions that may have nurtured these cultural forms.

In recent decades, numerous interdisciplinary studies have aimed to elucidate the relationships between climate, subsistence strategies, and cultural development (see, for instance, \cite{burdukiewicz2014origin}, \cite{chase1994symbols}, \cite{sterelny2021foragers}, \cite{ll2013cognition}, \cite{wollstonecroft2011investigating}, \cite{bradley2014past}, \cite{cook2020examining}). Many of these works are derived from the long-standing hypothesis that colder climates, offering large mammals as stable, high-calorie resources, could reduce the frequency of hunting events needed to sustain a community. In turn, fewer hunting events might create windows of leisure time. This additional time, no longer spent merely ensuring survival, could be invested in the production of symbolic objects, the performance of rituals, the refinement of storage methods, and the reinforcement of social bonds. Scholars have suggested that even subtle differences in the spoilage rates of stored food can shape the long-term evolutionary trajectory of societies, as lower spoilage reduces uncertainty and allows the accumulation of surplus food (see also \cite{clark2022domestic}, \cite{apicella2012social}).

However, existing narratives often treat technological and social management strategies as static or secondary. The complexity of social cooperation, the refinement of food preservation techniques, and the sophistication of storage containers and shelters are frequently acknowledged only as afterthoughts, rather than integral parts of the environmental-human nexus. Indeed, much earlier scholarship tended to regard "culture" as a byproduct of environmental abundance, rather than as an interactive and co-evolving component of the subsistence system itself. More recent discussions, influenced by theoretical and methodological advances in archaeology and cultural evolution studies, have emphasized that "culture" and "environment" are mutually constitutive. In other words, as resource management skills (such as improved storage techniques, effective preservation methods, or knowledge of seasonal resource distribution) develop, they not only mitigate spoilage and improve yields, but also reshape ecological constraints themselves (cf. \cite{kelly2013lifeways}, \cite{richerson2008not}, \cite{henrich2016secret}, \cite{zeder2012broad}, \cite{smith2011general}).

This study aims to situate itself at the intersection of these debates by offering a refined mathematical model coupled with a computational simulation. The core question guiding this research is: \emph{Does environmental stability, characterized by lower spoilage probability and higher yields, indirectly foster conditions more conducive to cultural complexity?} Our hypothesis does not rely on simple deterministic claims. Rather, we posit that, under certain assumptions, a group operating in a stable environment is likely to experience fewer constraints in food acquisition and preservation, thus freeing up valuable time. This additional time can then be allocated to cultural pursuits, be they artistic, ritualistic, or technological in nature. Over repeated generations, even modest reductions in spoilage or slight increases in yield can be amplified through feedback loops, allowing cultural complexity to take root and flourish.

To explore this hypothesis, we propose a formal model that captures the essential variables: yield, spoilage, resource management skill, and cultural complexity. We introduce a factor $G$ representing resource management ability (encompassing both technological and social know-how) and a factor $C$ representing cultural complexity. The effective yield per hunting event is modeled as increasing with both $G$ and $C$, capturing how cultural and technological sophistication can improve the subsistence returns. Meanwhile, daily spoilage $p$ continuously reduces stored food, making it a challenge. Agents (or “human groups”) attempt to maximize $C$ through a reinforcement learning (RL) algorithm \cite{mosavi2020comprehensive, souchleris2023reinforcement, abideen2021digital} that chooses whether to hunt, invest in resource management, or engage in cultural activities on a given day.

By running extensive simulations in which agents are exposed to various combinations of $Y$ (yield) and $p$ (spoilage), we can observe emergent patterns. While we cannot claim that these results constitute direct historical evidence, indeed, no simulation can replicate the full complexity of human prehistory, the patterns they reveal are suggestive. They indicate that stable, high-yield and low-spoilage environments lead agents to allocate less time to subsistence and more to cultural pursuits. Over time, this trend correlates with higher final values of $C$.

Unlike previous works that might have been limited to hand-waving claims or lacked formal rigor, we provide a tractable mathematical proof of the core proposition: under reasonable assumptions, a reduction in spoilage and an increase in yield reduce hunting frequency and thus increase available time for cultural elaboration. This theoretical result stands independent of any single simulation run. The simulations, in turn, serve as a form of empirical validation within the model’s own constraints, showing that the theoretical relationship holds even when multiple stochastic and dynamic factors are introduced.

It is important to emphasize that our approach does not claim that cold climates or stable conditions definitively caused cultural complexity. Rather, we argue that these environmental factors create a fertile ground, an enabling condition, within which cultural complexity can more easily emerge and intensify. Historical processes are always contingent, influenced by factors like social structures, cognitive capacities, migrations, and ecological disasters that we do not model here. Nonetheless, by showing a plausible causal chain, lower spoilage leads to fewer hunts, which leads to more free time, which in turn can be allocated to cultural activities, and we contribute to a growing body of literature that treats culture not as an isolated phenomenon, but as dynamically linked to environmental parameters.

The remainder of this paper will detail our modeling framework and computational experiments. In the Methods section, we outline the mathematical assumptions, the sets of actions, and the state variables. We then present theoretical propositions and a theorem that formalize the intuition behind our main claim. In the Experiments section, we describe the reinforcement learning setting and present the simulation results, supported by Ordinary Least Squares (OLS) regression analysis and data visualization. Finally, we reflect on the meaning of these findings for our understanding of early cultural complexity, acknowledging the limitations of any model that attempts to bridge deep prehistory and computational abstraction.

\section{Methods}

\subsection{Model Setup}
We consider a simplified ecological model in which a human group inhabits an environment characterized by two key parameters: daily yield $Y > 0$ from hunted fauna and daily spoilage probability $p \in [0,1]$ that reduces stored food each day. The group requires a fixed annual amount of food $F > 0$ and must allocate its time to hunting, investing in resource management skills, or engaging in cultural activities. Time is discretized into daily units on a fixed horizon $T$ days.

\begin{assumption}[Basic Requirements]
The group must secure a total amount of food $F$ over $T$ days. Each day, a fixed daily consumption $c > 0$ is subtracted from stored food. If the stored food becomes negative, the group suffers a starvation penalty.
\end{assumption}

\begin{assumption}[Actions]
On each day, the group selects one of the following actions:
\begin{enumerate}
\item \textbf{Hunt:} Costs $d_h > 0$ days (e.g., $d_h=2$). The group consumes daily rations during these days. After completing the hunt, they obtain an effective yield proportional to $Y$, adjusted by the group's management skill $G$ and cultural complexity $C$. The effective yield is:
\[
\text{eff}(Y,G,C) = Y \bigl(1 + 0.1G \bigr)\bigl(1 + 0.01C \bigr).
\]
\item \textbf{Invest (Resource Management):} Costs $d_i=1$ day. Improves $G$ by a fixed increment $\Delta G > 0$.
\item \textbf{Culture:} Costs $d_c=1$ day. Increases $C$ by a unit and grants a small immediate reward. 
\end{enumerate}
\end{assumption}

\begin{assumption}[Daily Spoilage]
At the end of each day, the stored food is reduced by a spoilage factor $p \in [0,1]$. If $f_{\text{current}}$ denotes current food, then:
\[
f_{\text{next}} = (f_{\text{current}} - c)(1 - p).
\]
If $f_{\text{next}} < 0$, the group incurs a large penalty and the process ends prematurely.
\end{assumption}

\begin{assumption}[Initial Conditions and Ending]
Initially, the group starts with a small positive amount of food (e.g., enough for several days). After $T$ days, the episode ends. We consider that no final additional reward is given, except that the final cultural complexity $C$ itself is recorded as a measure of success. Thus, the objective is to maximize $C$ without starving.
\end{assumption}

\subsection{Propositions and Theorem}

The following propositions formalize how lower spoilage ($p$) and higher yield ($Y$) affect time allocation and cultural complexity $C$. Lower $p$ implies that each hunted food unit remains available longer. Similarly, a larger $Y$ per hunt reduces the required hunting frequency.

\begin{assumption}[Comparative Conditions]
Consider two environments $A$ and $B$. In $A$, we have a higher yield $Y_A > Y_B$ and lower spoilage $p_A < p_B$. Both groups require $F$ units of food annually and have access to the same actions.
\end{assumption}

\begin{proposition}\label{prop:hunt}
Under the given assumptions, the environment $A$ requires fewer hunts to achieve the annual requirement $F$ due to $Y_A > Y_B$ and a lower effective daily spoilage.
\end{proposition}

\begin{proof}
If the hunt yield is higher and the spoilage is lower, the effective retained food per hunt in $A$ is strictly greater than in $B$. Hence, to obtain $F$, fewer hunts are necessary in $A$. Thus, $H_A < H_B$, where $H_i$ is the number of hunts in the environment $i$.
\end{proof}

\begin{proposition}\label{prop:time}
Because fewer hunts are required in $A$, the group in $A$ allocates fewer total days to hunting, resulting in more spare days available for either cultural activities or resource management.
\end{proposition}

\begin{proof}
If $H_A < H_B$ and each hunt costs at least $d_h$ days, then the total hunting time $T_{hunt,A} < T_{hunt,B}$. Since $T$ is fixed, $T_{\text{free},A} = T - T_{hunt,A} > T - T_{hunt,B} = T_{\text{free},B}$.
\end{proof}

\begin{theorem}[Cultural Complexity Advantage]\label{thm:main}
Let $C_i$ denote the cultural complexity achieved in environment $i$. If $H_A < H_B$ and the environment $A$ invest the extra free time in cultural actions, then $C_A > C_B$. Furthermore, since a higher $C$ increases the future hunting efficiency through $(1+0.01C)$, environment $A$ experiences a strengthening cycle that further elevates cultural complexity.
\end{theorem}

\begin{proof}
By Proposition \ref{prop:time}, environment $A$ has strictly more free days to allocate. Culture action increases $C$ by one unit per day and provides immediate rewards. Given that $A$ can afford more cultural actions (due to less frequent hunting), $C_A > C_B$. This increment in $C$ loops back to improve effective yield in future hunts, reinforcing the cycle and maintaining or increasing the gap.
\end{proof}

The theorem shows that under lower spoilage and higher yield conditions, groups can minimize the hunting frequency, thereby freeing up time for cultural activities. This leads to an emergent advantage in cultural complexity.

\section{Experiments}

\subsection{Simulation Procedure}
We conduct a computational experiment using a RL framework. Each agent represents a hypothetical human group operating within a given environmental condition. The agent selects actions from a discrete set (Hunt, Invest, Culture) each day to maximize cultural complexity $C$ while avoiding starvation. We simulate $N$ agents, each assigned a unique environment defined by a pair $(Y, p)$. The yield $Y$ and spoilage $p$ values differ among agents, allowing us to capture a wide range of ecological settings. We run multiple training episodes for each agent and evaluate their learned policies after training.

\subsection{Parameter Settings}
All agents share the same initial baseline conditions, including an initial food stock, a daily consumption rate, and a total number of days $T$. We vary $Y$ and $p$ between agents to cover a broad range of ecological conditions. After training is complete, we record the final cultural complexity $C$ achieved by each agent.

\subsection{Data Collection and Analysis}
We gather the final cultural complexity values $C$ from all agents and their respective $(Y, p)$ parameters. We apply an OLS regression to relate $C$ to $Y$ and $p$. We also generate plots to illustrate how $C$ correlates with $Y$ and $p$. Regression and visualization allow us to assess the agreement between computational outcomes and theoretical predictions.

\subsection{Reinforcement Learning Setting}
We formalize the theoretical model as a Markov decision process (MDP). The state at day $t$ includes:  
\begin{itemize}
\item Stored food level normalized by the annual requirement, $f_t/F$.
\item Resource management skill level $G$.
\item Cultural complexity $C$.
\item The fraction of remaining time, $(T-t)/T$.
\item Normalized yield, $Y/3000$, and spoilage probability $p$.
\end{itemize}

The agent chooses one action per day:
\begin{enumerate}
\item \textbf{Hunt:} Occupies $d_h=2$ days, each incurring daily consumption and spoilage. After completing the hunt, the agent obtains $\text{eff}(Y,G,C) = Y (1+0.1G)(1+0.01C)$ units of food.
\item \textbf{Invest:} Takes $1$ day. Improves $G$ by a fixed increment $\Delta G>0$. Daily consumption and spoilage apply.
\item \textbf{Culture:} Takes $1$ day. Increases $C$ by $1$ and gives a small immediate reward. Daily consumption and spoilage apply.
\end{enumerate}

The reward structure penalizes starvation. If $f_t$ falls below zero at the end of the day, the agent receives a $-100$ penalty and the episode ends.  Culture action gives a $+5$ immediate reward. No other action grants direct positive rewards. On the last day $T$, we record $C$ for analysis. Although $C$ is not added as a terminal reward, it represents the primary outcome.

\paragraph{RL Model and Optimization:}  
We implement an actor-critic RL model. The policy (actor) and value function (critic) are parameterized by neural networks with two hidden layers of 64 units each, using rectified linear unit (ReLU) activations. We orthogonally initialize all weight matrices to improve stability. The policy network outputs logits over the three possible actions, and the critic network outputs a scalar value estimate.

We employ an advantage-based policy gradient approach similar to A2C. Let $\pi_\theta(a|s)$ be the policy and $V_\phi(s)$ the value function. For a batch of collected episodes, we compute the advantages $A_t = R_t + \gamma V_\phi(s_{t+1}) - V_\phi(s_t)$, where $R_t$ is the immediate reward at time $t$ and $\gamma$ is the discount factor. We then update $\theta$ by ascending the gradient of the objective $J(\theta) = \mathbb{E}[\log \pi_\theta(a_t|s_t) A_t]$, and update $\phi$ by minimizing the mean-squared error between $V_\phi(s_t)$ and $R_t + \gamma V_\phi(s_{t+1})$. We use the Adam optimizer with a fixed learning rate and apply gradient clipping to maintain stable updates.

\paragraph{Training Details:}  
Each agent is trained for 50 episodes, each with $T=365$ days. In each episode, the agent attempts to secure enough food while allocating days to resource management or cultural activities. After training, we conducted 100 evaluation tests without learning updates and measured the mean $C$.

\paragraph{Parameter Variations:}  
We generate $N=1000$ agents, each assigned a random $(Y,p)$ from predefined ranges ($Y \in [1000,3000]$, $p \in [0.2,0.5]$). This large sample size provides sufficient variability to statistically analyze how changes in $Y$ and $p$ affect the results.

\begin{figure}
    \centering
    \includegraphics[width=0.7\linewidth]{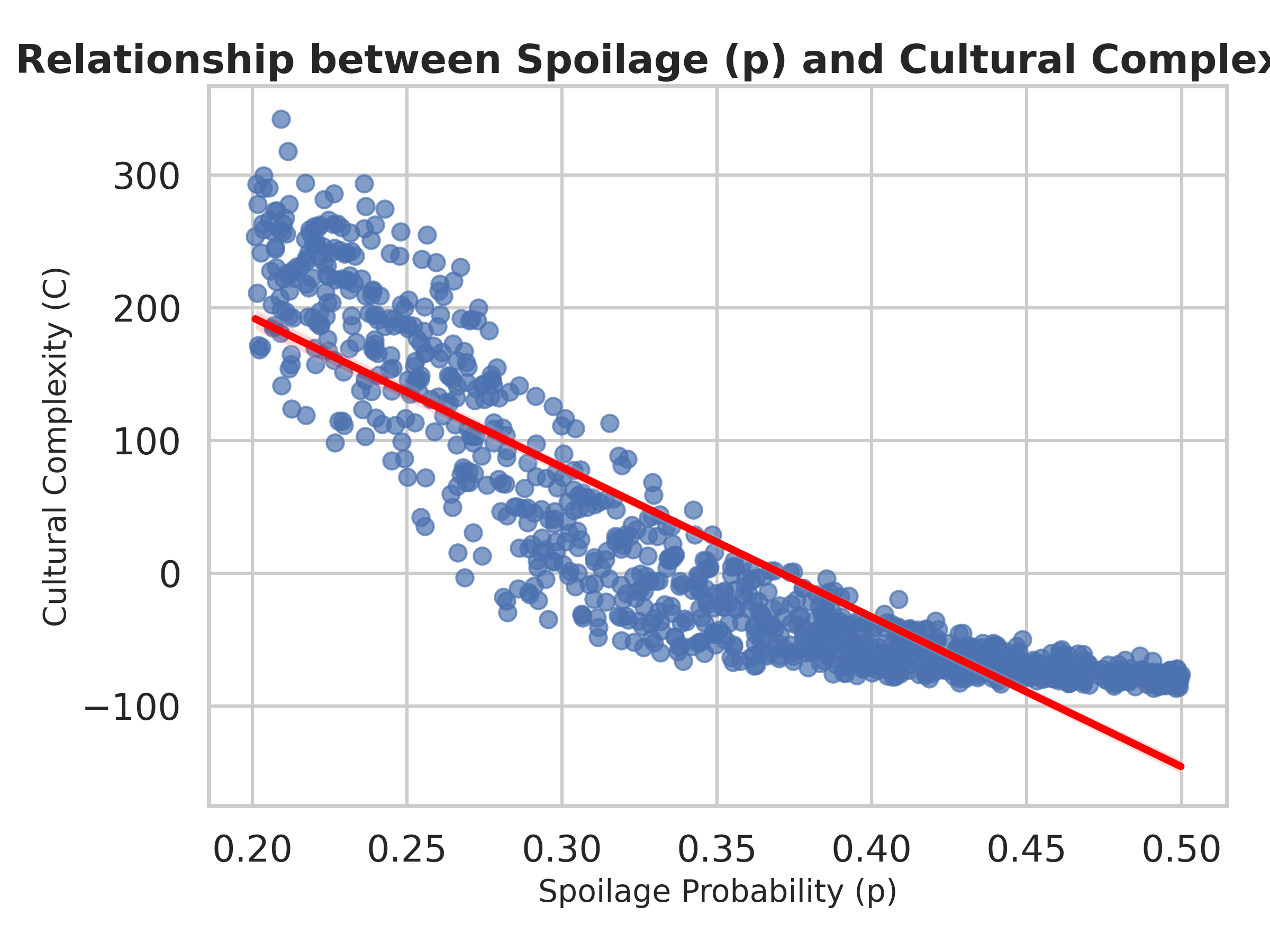}
    \caption{Cultural complexity $C$ as a function of spoilage probability $p$. Each point represents one agent. The red line is a linear fit. As $p$ increases, $C$ declines.}
    \label{fig:spoilage_vs_culture}
\end{figure}

\section{Results}

\subsection{Overall Patterns}
The results indicate a strong negative association between the probability of spoilage $p$ and cultural complexity $C$. Agents in environments with lower $p$ generally achieve higher $C$. While higher yield $Y$ also correlates with increased $C$, its effect is weaker than that of $p$. Figure~\ref{fig:spoilage_vs_culture} illustrates the steep decline in $C$ as $p$ increases. Figure~\ref{fig:yield_vs_culture} shows a gentle upward trend in $C$ as $Y$ grows.

\begin{figure}
    \centering
    \includegraphics[width=0.7\linewidth]{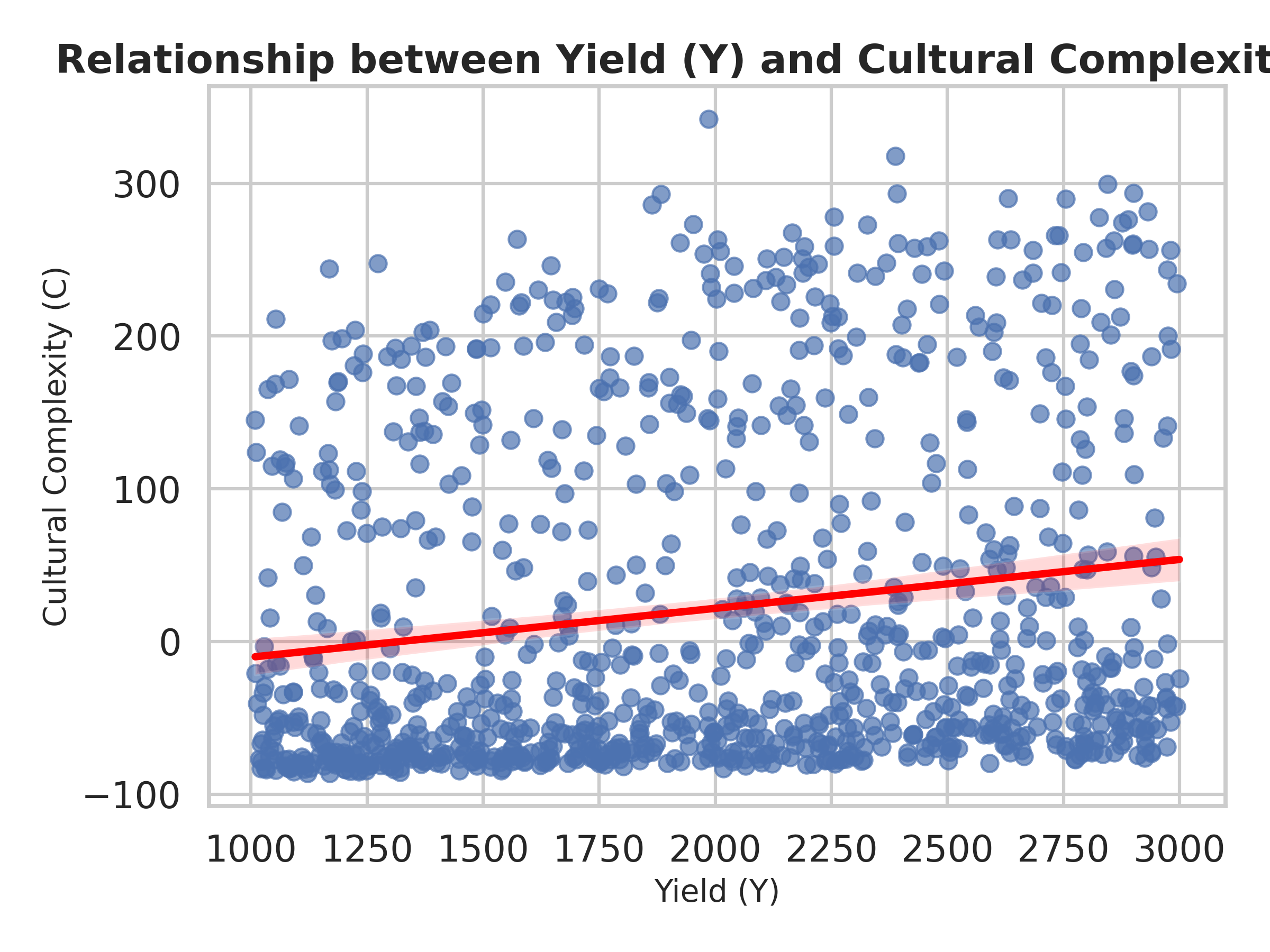}
    \caption{Cultural complexity $C$ as a function of yield $Y$. Each point represents one agent. The red line is a linear fit. As $Y$ increases, $C$ tends to rise modestly.}
    \label{fig:yield_vs_culture}
\end{figure}

\subsection{Statistical Relationships}
We fit an OLS regression model with $C$ as the dependent variable and $(Y,p)$ as independent variables:
\begin{table}[h!]
\centering
\begin{tabular}{lccc}
\hline
Parameter & Coefficient & p-value & StdErr \\
\hline
const & 347.45 & $<0.000001$ & 7.50 \\
$Y$ (x1) & 0.0370 & $<0.000001$ & 0.0024 \\
$p$ (x2) & -1134.83 & $<0.000001$ & 16.23 \\
\hline
\end{tabular}
\caption{OLS regression results. Both $Y$ and $p$ are highly significant. The negative coefficient for $p$ is large in magnitude, indicating that spoilage strongly reduces $C$.}
\label{tab:regression}
\end{table}

The regression (Table~\ref{tab:regression}) confirms that $p$ exerts a dominant negative influence on $C$. A small increase in $p$ leads to a large decrease in $C$. The effect of $Y$ is positive but less pronounced. These findings are consistent with the theoretical framework, where stable low-spoilage conditions reduce hunting frequency and allow for more cultural activities.

\begin{figure}
    \centering
    \includegraphics[width=0.7\linewidth]{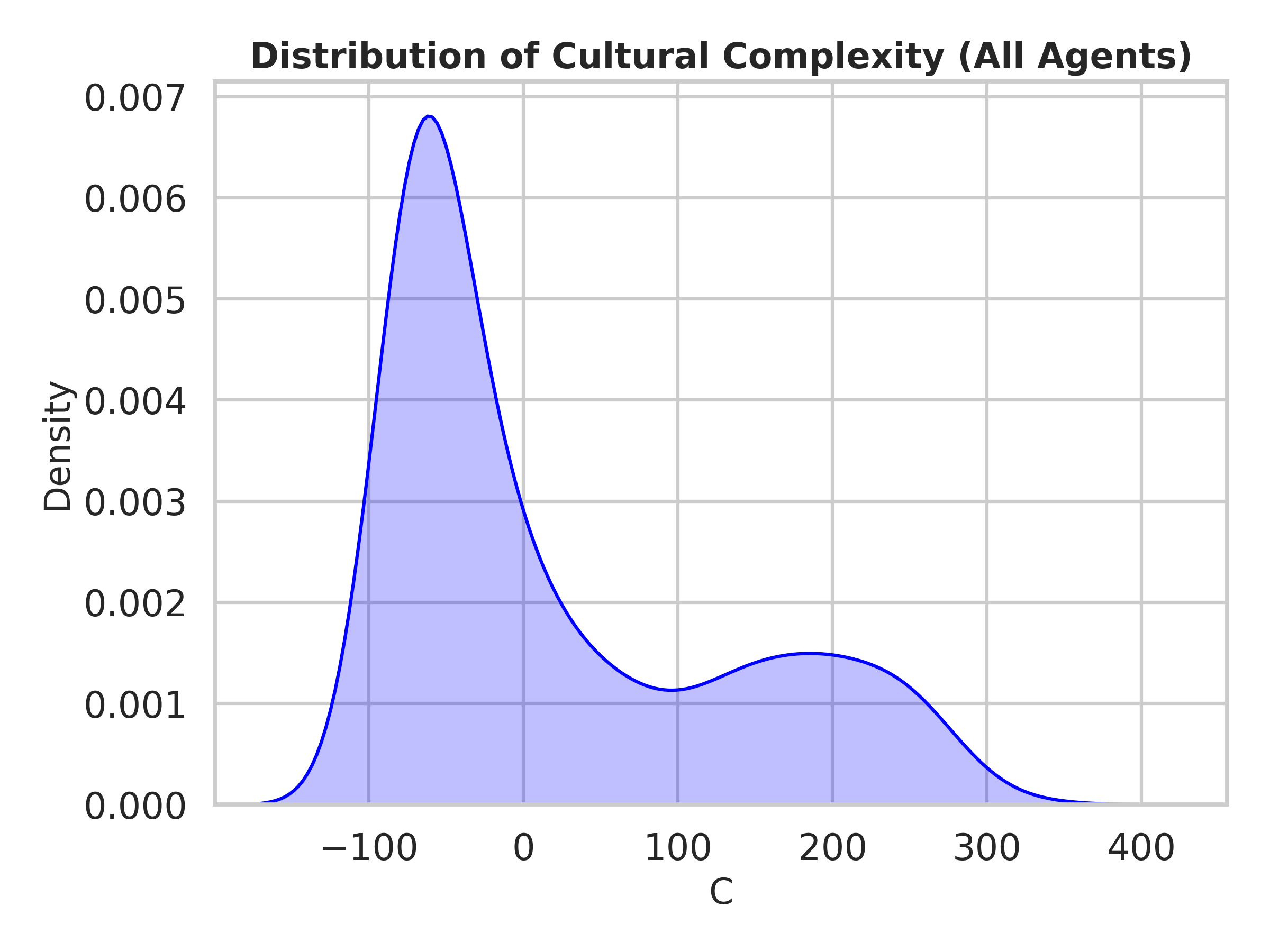}
    \caption{Distribution of cultural complexity $C$ across all agents. Most agents cluster near low or slightly negative $C$ values, but some achieve higher complexity.}
    \label{fig:distribution}
\end{figure}

Figure~\ref{fig:distribution} shows that many agents fail to achieve high $C$, reflecting frequent hunting and occasional starvation. However, some agents in stable environments reach substantially higher $C$. Figure~\ref{fig:heatmap} highlights how low $p$ combined with moderately high $Y$ fosters conditions conducive to cultural elaboration.

\begin{figure}
    \centering
    \includegraphics[width=0.7\linewidth]{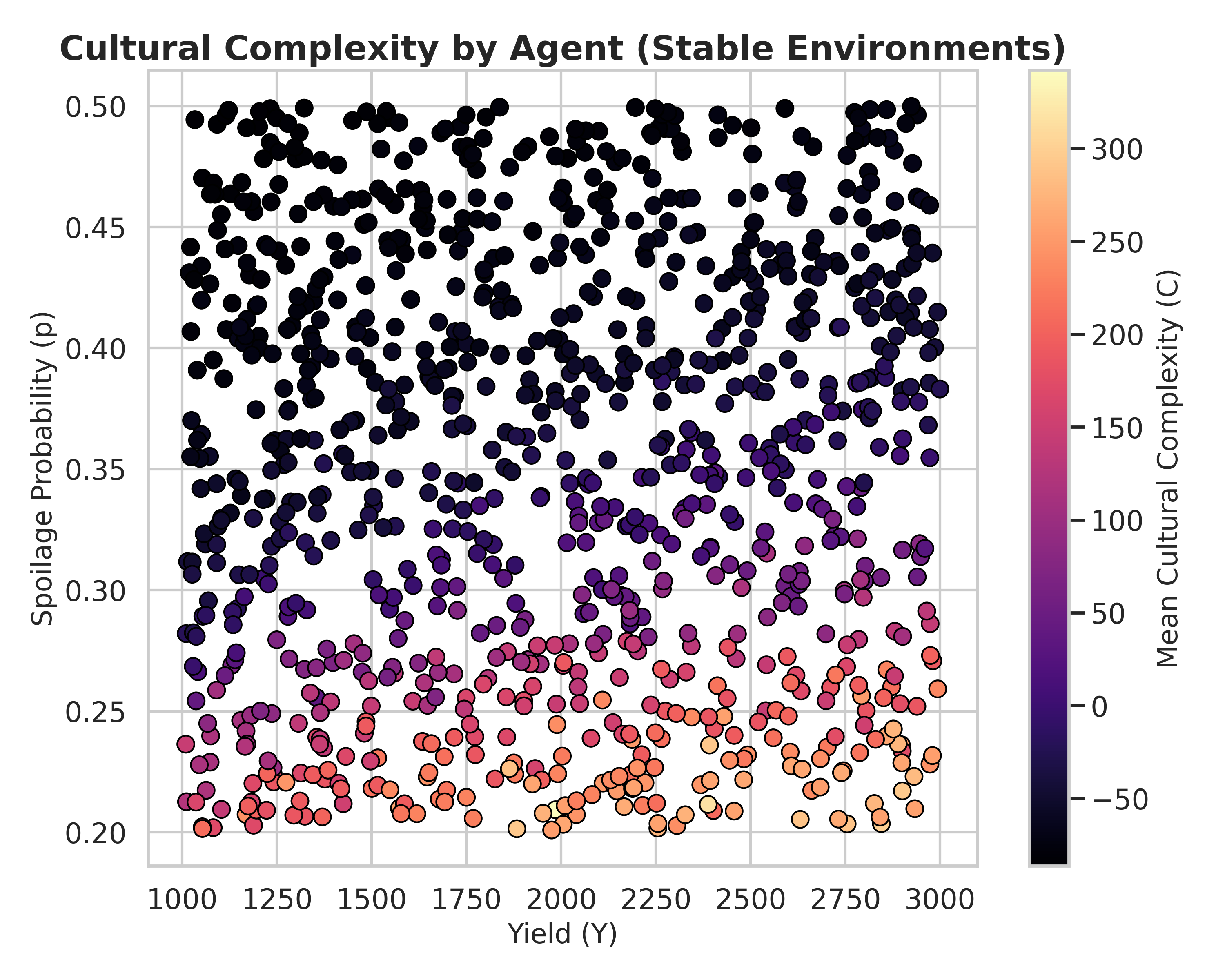}
    \caption{A two-dimensional map showing $C$ as a function of $Y$ (horizontal axis) and $p$ (vertical axis). Colors indicate mean $C$ achieved. Lower $p$ and higher $Y$ environments produce higher $C$.}
    \label{fig:heatmap}
\end{figure}

\section{Discussion}
Our findings align with theoretical arguments that environmental stability, represented by lower spoilage rates and adequate yields, can indirectly foster the conditions necessary for cultural complexity. This idea resonates with several strands of research in archaeology, anthropology, and cultural evolution.

Numerous studies have addressed the relationship between stable resource bases and cultural elaboration. For example, \citet{gamble2007origins} discusses how consistent resource availability in Pleistocene Europe could have facilitated symbolic behavior, while \citet{mithen1997prehistory} explores how cognitive fluidity and symbolic thought might have emerged when subsistence pressures diminished. Similarly, \citet{richerson2008not} emphasizes how cultural capacities can expand when ecological constraints are relaxed, and \citet{shennan2009pattern} links variation in subsistence strategies to the complexity of material culture.

In our model, the lower $p$ reduces the need for constant hunting, mirroring the arguments that stable conditions free humans from perpetual resource pursuit, allowing them to invest in activities unrelated to immediate survival. Authors such as \citet{klein2009human} and \citet{d2011origin} have noted that symbolic artifacts, ritual practices, and social cooperation intensify when groups face fewer environmental risks. Our computational results lend quantitative support to these claims. While these simulations do not replicate actual human prehistory, they suggest a plausible mechanism by which environmental parameters shape the space of possible cultural outcomes.

Another line of research, represented by works like \citet{boyd1988culture} and \citet{henrich2016secret}, highlights that cumulative cultural evolution thrives under conditions where knowledge transmission is reliable and long-term planning is possible. Our simulation shows that when daily spoilage is less severe, agents can plan beyond mere survival. This shift enables them to develop resource management ($G$) and engage in cultural activities ($C$) that would be impossible under harsher conditions. The model indicates that as groups improve their resource management skills, the returns on hunting days improve, reinforcing the cycle of cultural investment. This insight echoes \citet{laland2017darwin} and \citet{maryanski2013secret}, who argue that material and social factors co-evolve, each facilitating the complexity of the other.

Furthermore, our results link environmental constraints to the trajectories of cultural complexity in a manner consistent with \citet{powell2009late}, who proposed that demographic and ecological factors together shape the rates of cultural innovation. In stable, low-spoilage environments, agents have fewer hunts, more time to learn and transmit skills, and greater opportunities to invest in non-survival activities. Although we do not model demographic factors directly, the pattern of improved cultural complexity in stable environments suggests that demographic expansions could further enhance these effects, as larger group sizes may sustain more complex cultural repertoires.

While our approach is abstract, the implications resonate with arguments about environmental influences on cultural emergence. Historical societies that managed to store and preserve food, such as those discussed by \citet{wrangham2009catching} and \citet{prentice2009cultural}, often saw increases in social complexity and specialization. The emphasis of our model on daily spoilage captures a simplified version of these preservation challenges. When preservation is easier, as represented by lower $p$, agents can afford to experiment with cultural activities.

It is important to note that we do not claim that cold climates or stable ecosystems directly caused cultural complexity. Rather, we suggest that such conditions reduce the time spent on basic subsistence and thus create opportunities for cultural elaboration. The notion that the environment sets the stage, rather than determines the script, is aligned with the work of \citet{clarke1978analytical} and \citet{bettinger2015hunter}, who stress that the environment provides constraints and opportunities that societies navigate through cultural and technological strategies.

\section{Conclusion}
The emergence of cultural complexity in prehistoric foraging societies remains a central question in the humanities, prompting inquiries that span archaeology, anthropology, and beyond. These investigations have often highlighted the potential significance of environmental stability, including factors such as resource predictability and reduced spoilage, as enabling conditions that opened the temporal and cognitive space for symbolic, ritualistic, and technological innovation.

Our study contributes to these discussions by providing a formal and computationally supported model that captures the interplay between yield, spoilage, subsistence labor, and cultural investment. By demonstrating a quantitative link between ecological parameters and cultural outcomes, this work does not assert a deterministic cause-and-effect scenario. Instead, it proposes that certain ecological settings created more favorable grounds for cultural elaboration. These settings, characterized by relatively stable conditions, allowed human groups to invest fewer days in securing immediate sustenance and more days in pursuits that, over time, contributed to cultural complexity.

This approach resonates with long-standing debates in the humanities, where scholars have argued that cultural forms emerge not merely as epiphenomena of abundance, but as adaptive and co-evolving strategies shaped by environmental constraints. Our findings offer a structured, testable hypothesis: that lower spoilage probabilities and moderately increased yields are not just background variables but essential components that shape the tempo and mode of cultural evolution. Through computational simulations and theoretical proofs, we have illustrated how groups inhabiting environments with certain ecological parameters could, in principle, shift their time budgets away from relentless foraging toward activities that nurture symbolic behavior, technological refinement, and social cooperation.

In acknowledging these patterns, we do not discount other vital factors, social hierarchy, cognitive development, language complexity, demographic processes, or the role of migration and intergroup exchange. Rather, we present one segment of a larger tapestry in which environmental stability stands out as a significant, if indirect, determinant of cultural flourishing. Future research can extend this model to incorporate dynamic ecological cycles, fluctuating resource patterns, and culturally transmitted practices. Such integration would further align computational modeling with the rich empirical record of the humanities, potentially revealing deeper insights into how and why human culture assumed its remarkable diversity and complexity over the long arc of prehistory.

\vspace{1em}

\bibliography{output}
\bibliographystyle{plainnat}
\end{document}